\documentclass[11pt]{llncs}

\usepackage{color}
\usepackage{latexsym}
\usepackage{amsmath}
\usepackage{amssymb}
\usepackage{cite}
\usepackage[noend,ruled,linesnumbered]{algorithm2e}
\usepackage{graphicx}

\newtheorem{fact}[theorem]{Fact}

\newcommand{\RN}{{Radio Network}\xspace}
\newcommand{\RNs}{{Radio Networks}\xspace}

\newcommand{\SNs}{{Sensor Networks}\xspace}

\newcommand{\KS}{$k$-selection\xspace}
\newcommand{\sKS}{static \KS}
\newcommand{\AB}{all-broadcast\xspace}

\newcommand{\OFA}{\textsc{One-fail Adaptive}\xspace}
\newcommand{\LFA}{\textsc{Log-fails Adaptive}\xspace}
\newcommand{\EBOBO}{\textsc{Exp Back-on/Back-off}\xspace}
\newcommand{\LLIBO}{\textsc{Loglog-iterated Back-off}\xspace}

\begin{document}

\mainmatter              % start of the contributions
%---------------------------------------------------------
%\title{Unbounded Contention Resolution:\\Selection in \RNs}
\title{Unbounded Contention Resolution in\\ Multiple-Access Channels
\thanks{A brief announcement of this work has been presented at PODC 2011. This work is supported in part by the Comunidad de Madrid grant S2009TIC-1692, Spanish MICINN grant TIN2008--06735-C02-01,
and National Science Foundation grant CCF-0937829.}}
\titlerunning{}  % abbreviated title (for running head)
\author{
   Antonio Fern\'andez Anta\inst{1}
   \and
   Miguel~A.~Mosteiro\inst{2,3}\\
   \and
   Jorge Ram\'on Mu\~noz\inst{3}
}
\authorrunning{A. Fern\'andez Anta, M. A. Mosteiro, and J. R. Mu\~noz}   % abbreviated author list (for running head)
\institute{
Institute IMDEA Networks, Madrid, Spain\\
\email{antonio.fernandez@imdea.org}
\and
Department of Computer Science, Rutgers University, Piscataway, NJ, USA\\ 
\email{mosteiro@cs.rutgers.edu}
\and
LADyR, GSyC, Universidad Rey Juan Carlos, M\'ostoles, Spain\\
\email{jorge.ramon@madrimasd.net}
} 

\maketitle              % typeset the title of the contribution

\newcommand{\footnotenonumber}[1]{{\def\thempfn{}\footnotetext{#1}}}

%%%%%%%%%%%%%%%%%%%%%%%%%%%%%%%%%%%%%%%%%%%%%%%%%%%%%%%%%%%%%%%%%%%%%%%%%%

\begin{abstract} 
A frequent problem in settings where a unique resource must be shared among users is how to resolve the contention that arises when all of them must use it, but the resource allows only for one user each time.
The application of efficient solutions for this problem spans a myriad of settings such as radio communication networks or databases.
\textcolor{black}{
For the case where the number of users is unknown, 
%but fixed, 
recent work has yielded fruitful results for local area networks and radio networks, although either 
a (possibly loose) upper bound on the number of users needs to be known~\cite{FM:kSelJournal},
or the solution is suboptimal~\cite{bendercontention},
or it is only implicit~\cite{GL:fattrees} or embedded~\cite{FCM:gossiping} in other problems, with bounds proved only asymptotically. 
In this paper, under the assumption that collision detection or information on the number of contenders is not available, 
we present a novel protocol for contention resolution in radio networks, and we recreate a protocol previously used for other problems~\cite{GL:fattrees, FCM:gossiping}, tailoring the constants for our needs.
In contrast with previous work, both protocols are proved to be optimal up to a small constant factor and with high probability for big enough number of contenders.
Additionally, the protocols are evaluated and contrasted with the previous work by extensive simulations. The evaluation shows that the complexity bounds obtained by the analysis are rather tight, and that both protocols proposed
have small and predictable complexity for many system sizes (unlike previous proposals).
}
\end{abstract}

%\vspace*{0.5cm}
%{\bf Corresponding author:}  Miguel A. Mosteiro \\
%{\bf Postal address:} 	 Rutgers University, Computer Science Department, 110 Frelinghuysen Rd., Piscataway, NJ 08854, USA\\
%{\bf Phone:}  +1 732-445-2001 x9591\\
%{\bf E-Mail: }  mosteiro@cs.rutgers.edu\\
%\vspace*{0.5cm}
%\vspace*{0.5cm}
%\begin{center}
%{\bf Please consider this submission as a regular paper.\\
%This paper should be considered for the best student-paper award, Jorge Ram\'on Mu\~noz is a full-time student.}
%\end{center}
%
%\thispagestyle{empty}
%\setcounter{page}{0} \clearpage %\pagestyle{plain}
%
%\newpage

\section{Introduction}
\label{section:intro}

The topic of this work is the resolution of contention in settings where an unknown number of users must access a single shared resource, but multiple simultaneous accesses are not feasible.
The scope of interest in this problem is wide, ranging from radio and local area networks to databases and transactional memory. (See~\cite{bendercontention} and the references therein.) 

A common theme in protocols used for this problem is the adaptive adjustment of some user variable that reflects its eagerness in trying to access the shared resource. Examples of such variable are the probability of transmitting a message in a radio network or the frequency of packet transmission in a local area network.
When such adjustment reduces (resp. increases) the contention, the technique is called \emph{back-off} (resp. \emph{back-on}). Combination of both methods are called \emph{back-on/back-off}. Protocols used may be further characterized by the rate of adjustment. E.g., \emph{exponential back-off}, \emph{polynomial back-on}, etc.
In particular, exponential back-off is widely used and it has proven to be efficient in practical applications where statistical arrival of contenders is expected. Nevertheless, worst case arrival patterns, such as bursty or \emph{batched} arrivals, are frequent~\cite{L:etherTraff,G:etherTraff}. 

A technique called \emph{\LLIBO} was shown to be within a sublogarithmic factor from optimal with high probability in~\cite{bendercontention}.~\footnote{For $k$ contenders, we define \emph{with high probability} to mean with probability at least $1-1/k^c$ for some constant $c>0$.}
The protocol was presented in the context of packet contention resolution in local area networks for batched arrivals. Later on, also for batched arrivals, we presented a back-on/back-off protocol in~\cite{FM:kSelJournal}, instantiated in the \KS problem in \RNs (defined in Section~\ref{section:prelim}). The latter protocol, named here \emph{\LFA,} is asymptotically optimal for any significant probability of error, but additionally requires that some upper bound (possibly loose) on the number of contenders is known. In the present paper, we remove such requirement. 
\textcolor{black}{
In particular, we present and analyze a protocol that we call \emph{\OFA} for \KS in \RNs. 
We also recreate and analyze another protocol for \KS, called here \emph{\EBOBO}, which was previously embedded in protocols for other problems and analyzed only asymptotically~\cite{GL:fattrees, FCM:gossiping}.
}
\textcolor{black}{Our analysis shows that  
\OFA and \EBOBO, both of independent interest,} resolve contention among an unknown and unbounded~\footnote{We use the term unbounded to reflect that not even an upper bound on the number of contenders is known. This should not be confused with the infinitely-many users model where there are countably infinitely many stations.~\cite{chlebusRNsurvey}} number of contenders with high probability in optimal time up to constants. 
Additionally, by means of simulations, we evaluate and contrast the average performance of all four protocols. The simulations show that the complexity bounds obtained in the analysis (with high probability) for these protocols are rather tight \textcolor{black}{for the input sizes considered.} Additionally, they show that they are faster that \LLIBO and more predictable for all network sizes than \LFA. 

\paragraph{Roadmap}
The rest of the paper is organized as follows. 
In the following section the problem, model, related work and results are detailed. 
\textcolor{black}{
In Section~\ref{section:ofa}, we introduce \OFA and its analysis.
\EBOBO is detailed and analyzed in Section~\ref{section:ebobo}.
}
The results of the empirical contrast of all four protocols is given in Section~\ref{section:eval} and we finish with concluding remarks and open problems in Section~\ref{section:conclude}. 

\section{Preliminaries}
\label{section:prelim}

A well-studied example of unique-resource contention is the problem of broadcasting information in a multiple-access channel. A multiple-access channel is a synchronous system that allows a message to be delivered to many recipients at the same time using a channel of communication but, due to the shared nature of the channel, the simultaneous introduction of messages from multiple sources produce a conflict that precludes any message from being delivered to any recipient. The particular model of multiple-access channel we consider here is the \RN, a model of communication network where the channel is contended (even if radio communication is not actually used~\cite{chlebusRNsurvey}). We first precise our model of \RN as follows.

%--------------------------------------------------------------------------------------------
\paragraph{The Model:} 
We consider a \RN comprised of {$n$} stations called \emph{nodes}. 
Each node is assumed to be potentially reachable from any other node in one communication step, hence, the network is characterized as \emph{single-hop} or \emph{one-hop} indistinctively.
Before running the protocol, nodes have no information, not even the number of nodes $n$ or their own label.
Time is supposed to be slotted in \emph{communication steps}. 
Assuming that the computation time-cost is negligible in comparison with the communication time-cost, time efficiency  is studied in terms of communication steps only. 
The piece of information assigned to a node in order to deliver it to other nodes is called a \emph{message}. 
The assignment of a message is due to an external agent and such an event is called a \emph{message arrival}.
Communication among nodes is {carried out by means of} radio broadcast on a shared channel.
If exactly one node transmits at a communication step, such a transmission is called \emph{successful} or \emph{non-colliding}, we say that the message was \emph{delivered}, and all other nodes \emph{receive} such a message. 
If more than one message is transmitted at the same time, a \emph{collision} occurs, the messages are garbled, and nodes only receive \emph{interference noise}. If no message is transmitted in a communication step, nodes receive only \emph{background noise}. 
In this work, nodes can not distinguish between interference noise and background noise, thus, the channel is called \emph{without collision detection}.
Each node is in one of two states, \emph{active} if it holds a message to deliver,  or \emph{idle} otherwise.
As in~\cite{bendercontention,K:selection,GL:fattrees}, we assume that a node becomes idle upon delivering its message, for instance when 
an explicit acknowledgement is received (like in the IEEE 802.11 Medium Access Control protocol~\cite{BMJ:contentionWindow}). For settings where the channel does not provide such functionality, 
%(e.g., a base station acknowledgement), 
such as \SNs, a hierarchical infrastructure may be predefined to achieve it~\cite{FCM:gossiping}, or a leader can be elected as the 
node responsible for acknowledging successful transmissions \cite{nakano:leader}.

One of the problems that require contention resolution in \RNs is the problem known in the literature as \emph{\AB}~\cite{chlebusRNsurvey}, or \emph{\KS}~\cite{K:selection}. In \KS, a set of $k$ out of $n$ network nodes have to access a unique shared channel of communication, each of them at least once. 
As in~\cite{bendercontention,K:selection,GL:fattrees}, in this paper we study \KS when all messages arrive simultaneously, or in a \emph{batch}. Under this assumption the \KS problem is called \emph{static}. 
A \emph{dynamic} counterpart where messages arrive at different times was also studied~\cite{K:selection}.

%--------------------------------------------------------------------------------------------
\paragraph{The Problem:}
Given a \RN where $k$ network nodes are activated by a message that arrives simultaneously to all of them, the \emph{\sKS} problem is solved when each node has delivered its message.

%--------------------------------------------------------------------------------------------
\paragraph{Related Work:}

\textcolor{black}{A number of fruitful results for contention resolution have been obtained assuming availability of collision detection.}
Martel presented in~\cite{M:Selection} a randomized adaptive protocol for $k$-Selection that works in $O(k+\log n)$ time in expectation\footnote{Througout this paper, $\log$ means $\log_2$ unless otherwise stated.}. As argued by Kowalski in~\cite{K:selection}, this protocol can be improved to $O(k+\log\log n)$ in expectation using Willard's expected $O(\log\log n)$ selection protocol of~\cite{W:logselection}. In the same paper, Willard shows that, for any given protocol, there exists a choice of $k\leq n$ such that selection takes $\Omega(\log\log n)$ expected time for the class of fair selection protocols (i.e., protocols where all nodes use the same probability of transmission to transmit in any given time slot). For the case in which $n$ is not known, in the same paper a $O(\log\log k)$ expected time selection protocol is described, again, making use of collision detection. 
If collision detection is not available, using the techniques of Kushilevitz and Mansour in~\cite{KM:broadcast}, it can be shown that, for any given protocol, there exists a choice of $k\leq n$ such that $\Omega(\log n)$ is a lower bound in the expected time to get even the first message delivered.

Regarding deterministic solutions, the $k$-Selection problem was shown to be in $O(k\log(n/k))$ already in the 70's by giving adaptive protocols that make use of collision detection~\cite{C:treealgpackbroad,H:adaplocaldist,TM:freePackAcc}. In all these results the algorithmic technique, known as \emph{tree algorithms}, relies on modeling the protocol as a complete binary tree where the messages are placed at the leaves. Later, Greenberg and Winograd~\cite{GW:determSelection} showed a lower bound for that class of protocols of $\Omega(k\log_k n)$. Regarding oblivious algorithms, {Koml\`os and Greenberg~\cite{KG:Selection}} showed the existence of $O(k\log(n/k))$ solutions even without collision detection but requiring knowledge of $k$ and $n$. More recently, Clementi, Monti, and Silvestri~\cite{CMS:selectiveFamilies} showed a lower bound {of $\Omega(k\log(n/k))$}, which also holds for adaptive algorithms if collision detection is not available.
In~\cite{K:selection}, Kowalski presented the construction of an oblivious deterministic protocol that, using the explicit selectors of Indyk~\cite{I:selectors}, gives a $O(k \mathop{\mathrm{polylog}} n)$ upper bound without collision detection.

In~\cite{GGT:controltower}, Ger\`eb-Graus and Tsantilas presented an
algorithm that solves the problem of realizing arbitrary $h$-relations in an $n$-node network, with probability at least $1-1/n^c,c>0$, in $\Theta(h+\log n\log\log n)$ steps. In an $h$-relation, each processor is the source as well as the destination of at most $h$ messages. Making $h=k$ this protocol can be used to solve \sKS. However, it requires that nodes know $h$.

\textcolor{black}{
Extending previous work on tree algorithms, Greenberg and Leiserson~\cite{GL:fattrees} presented randomized routing strategies in fat-trees for bounded number of messages. Underlying their algorithm lies a sawtooth technique used to ``guess'' the appropriate value for some critical parameter (load factor), that can be used to ``guess'' the number of contenders in \sKS. Furthermore, modulo choosing the appropriate constants, \EBOBO uses the same sawtooth technique. Their algorithm uses $n$ and it is analyzed only asymptotically.
}

Monotonic back-off strategies for contention resolution of batched arrivals of $k$ packets on simple multiple access channels, a problem that can be seen as \sKS, have been analyzed in~\cite{bendercontention}. In that paper, it is shown that $r$-\emph{exponential back-off}, a monotonic technique used widely that has proven to be efficient for many practical applications is in $\Theta(k \log^{\log r} k)$ for batched arrivals. The best strategy shown 
%in the same paper 
is the so-called \emph{loglog-iterated back-off} with a makespan in $\Theta(k\log\log k/\log\log\log k)$ with probability at least $1-1/k^c,c>0$, which does not use any knowledge of $k$ or $n$.
\textcolor{black}{
In the same paper, the sawtooth technique used in~\cite{GL:fattrees} is informally described in a paragraph while pointing out that it yields linear time for contention resolution thanks to non-monotonicity, but no analysis is provided.
}

\textcolor{black}{
Later on, Farach-Colton and Mosteiro presented an optimal protocol for Gossiping in \RNs in~\cite{FCM:gossiping}. The sawtooth technique embedded in~\cite{GL:fattrees} is used in that paper as a subroutine to resolve contention in linear time as in \EBOBO. However, the algorithm makes use of $n$ to achieve the desired probability of success and the analysis is only asymptotical. 
}

A randomized adaptive protocol for \sKS in a one-hop \RN without collision detection was presented in~\cite{FM:kSelJournal}. 
The protocol is shown to solve the problem in $(e+1+\xi)k+O(\log^2(1/\varepsilon))$ steps with probability at least $1-2\varepsilon$, where $\xi>0$ is an arbitrarily small constant and $0<\varepsilon\leq 1/(n+1)$. 
Modulo a constant factor, the protocol is optimal if $\varepsilon\in\Omega(2^{-\sqrt{n}})$. 
However, the algorithm makes use of the value of $\varepsilon$, which must be upper bounded as above in order to guarantee the running time. Therefore, knowledge of $n$ is required.

%--------------------------------------------------------------------------------------------
\paragraph{Our Results:}
\textcolor{black}{
In this paper, we present a novel randomized protocol for \sKS in a one-hop \RN, and we recreate a previously used technique suiting the constants for our purpose and without making use of $n$.
Both protocols work without collision detection and do not require information about the number of contenders.
As mentioned, these protocols are called \OFA and \EBOBO.
It is proved that \OFA solves \sKS within $2(\delta+1)k+O(\log^2 k)$ steps, with probability at least $1-2/(1+k)$, for ${\mathrm e}<\delta\leq\sum_{j=1}^{5}(5/6)^j$.
On the other hand, \EBOBO is shown to solve \sKS within $4(1+1/\delta)k$ steps with probability at least $1-1/k^c$ for some constant $c>0$, $0< \delta < 1/{\mathrm e}$, and big enough $k$.
}
Given that $k$ is a lower bound for this problem, both protocols are optimal (modulo a small constant factor)
\textcolor{black}{for big enough number of contenders.}

Observe that the bounds and the probabilities obtained are given as functions of the parameter $k$, as done in~\cite{bendercontention}, since this is the input parameter of our version of the problem. 
A fair comparison with the results obtained as function of $k$ and $n$ would require that $k$ is large enough, so that $n=\Omega(k^c)$, for some constant $c$.
\textcolor{black}{Both protocols presented are of interest because,} although protocol \EBOBO is simpler, \OFA achieves a better multiplicative factor, although the constant in the sublinear additive factor may be big for small values of $k$.

Additionally, results of the evaluation by simulation of the average behavior of \OFA and \EBOBO and a comparison with \LFA and \LLIBO are presented.
Both algorithms \OFA and \EBOBO run faster than \LLIBO on average, even for small values of $k$. Although \LLIBO has higher asymptotic complexity, one may have expected
that it may run fast for small networks. On the other hand, the knowledge on a bound of $k$ assumed by \LFA seems to provide an edge with respect to
\OFA and \EBOBO for large values of $k$. However, \LFA has a much worse behavior than the proposed protocols for small  to moderate network sizes ($k \leq 10^5$).
In any case, for all values of $k$ simulated, \OFA and \EBOBO have a very stable and efficient behavior.

%--------------------------------------------------------------------------------------------

\section{\OFA}
\label{section:ofa}
As in \LFA~\cite{FM:kSelJournal}, \OFA is composed by two interleaved randomized algorithms, each intended to handle the communication for different levels of contention.
One of the algorithms, which we call \emph{AT}, is intended for delivering messages while the number of nodes contending for the channel is above some logarithmic threshold (to be defined later). The other algorithm, called \emph{BT}, has the purpose of handling message delivery after that number is below that threshold. Nonetheless, a node may transmit using the BT (resp. AT) algorithm even if the number of messages left to deliver is above (resp. below) that threshold. 

Both algorithms, AT and BT, are based on transmission trials with certain probability and what distinguishes them is just the specific probability value used. It is precisely the particular values of probability used in each algorithm what differentiates \OFA from \LFA.
For the BT algorithm, the probability of transmission is inversely logarithmic on the number of messages already transmitted, while in \LFA that probability was fixed.
For the AT algorithm the probability of transmission is the inverse of an estimation on the number of messages left to deliver. 
In \OFA this estimation is updated continuously, whereas in \LFA it was updated after some steps without communication. 
These changes yield a protocol still linear, but now it is not necessary to know $n$.
Further details can be seen in Algorithm~\ref{alg}.

For clarity, Algorithms AT and BT are analyzed separately taking into account in both analyses the presence of the other. We show the efficiency of the AT algorithm in producing successful transmissions while the number of messages left is above some logarithmic threshold, and the efficiency of the BT algorithm handling the communication after that threshold is crossed. For the latter, we use standard probability computations to show our time upper bound. 
For the AT algorithm, we use concentration bounds to show that the messages are delivered with large enough probability, while the density estimator $\widetilde{\kappa}$ does not exceed the actual number of messages left.
This second proof is more involved since it requires some preliminary lemmas.
We establish here the main theorem, which is direct consequence of the lemmata described 
%that can be found in~\cite{FM:kSelArxiv}.
that can be found in the Appendix.

%Regarding the analysis, we establish here the main theorem, which is direct consequence of the lemmata left to the Appendix for brevity.
\begin{theorem}
For any $e<\delta\leq\sum_{j=1}^{5}(5/6)^j$ and for any one-hop \RN under the model detailed in Section~\ref{section:intro}, \OFA solves \sKS within $2(\delta+1)k+O(\log^2 k)$ communication steps, with probability at least $1-2/(1+k)$.
\end{theorem}

\begin{algorithm}[tbp]
\label{alg}
\caption{\OFA.
Pseudocode for node $x$. $\delta$ is a constant such that $e<\delta\leq\sum_{j=1}^{5}(5/6)^{j}$.
}
\small
\SetKwData{Step}{communication-step}
\SetKwData{Mess}{message}
\SetKwFor{Upon}{upon}{do}{endupon}
\SetKwFor{On}{on}{do}{endon}
\SetKwFor{Task}{Task}{}{endtask}
\SetKwFor{WP}{}{do}{enddo}
\dontprintsemicolon
%\BlankLine\;
%\BlankLine\;
\Upon{$\Mess$ arrival}{
   $\widetilde{\kappa}\leftarrow \delta+1$\label{init}\tcp*[f]{Density estimator}\;
   $\sigma\leftarrow 0$\tcp*[f]{Messages-received counter}\;
   \textbf{start} tasks $1, 2$ and $3$\;
}
%\BlankLine\;
%\BlankLine\;
\Task{1}{
\ForEach{$\Step =1,2,\dots$}{
   \If(\tcp*[f]{BT-step}){$\Step\equiv 0\pmod{2}$}{
      transmit $\langle x, \Mess\rangle$ with prob $1/(1+\log (\sigma+1))$\;
   }
   \Else(\tcp*[f]{AT-step}){
         transmit $\langle x, \Mess\rangle$ with probability $1/\widetilde{\kappa}$\; 
         $\widetilde{\kappa}\leftarrow \widetilde{\kappa}+1$\;
   }
}
}
%\BlankLine\;
%\BlankLine\;
\Task{2}{
   \Upon{reception from other node}{
      $\sigma\leftarrow \sigma+1$\;
      \If(\tcp*[f]{BT-step}){$\Step\equiv 0\pmod{2}$}{
         $\widetilde{\kappa}\leftarrow \max\{\widetilde{\kappa}-\delta,\delta+1\}$\;
      }
      \Else(\tcp*[f]{AT-step}){
         $\widetilde{\kappa}\leftarrow \max\{\widetilde{\kappa}-\delta-1,\delta+1\}$\;
      }
   }
}
%\BlankLine\;
%\BlankLine\;
\Task{3}{
   \textbf{upon} $\Mess$ $delivery$ \textbf{stop}\;
}
\end{algorithm}

\section{\EBOBO}
\label{section:ebobo}
The algorithm presented in this section is based in contention windows. 
That is, each node repeatedly chooses uniformly one time slot within an interval, or \emph{window}, of
time slots to transmit its message. 
Regarding the size of such window, our protocol follows a back-on/back-off strategy. 
Namely, the window is increased in an outer loop and decreased in an inner loop, as detailed in Algorithm~\ref{algo}. 

\begin{algorithm}[t]
\dontprintsemicolon
\label{algo}
\caption{Window size adjustment in \EBOBO. $0<\delta<1/\mathrm{e}$ is a constant.}
  \For{$i=\{1,2,\dots\}$}{
      $w\leftarrow 2^i$\;
      \While{$w\geq 1$}{ 
          Choose uniformly a step within the next $w$ steps\;
          $w\leftarrow w\cdot(1-\delta)$\;
      } 
  }
\end{algorithm}

The intuition for the algorithm is as follows. 
Let $m$ be the number of messages left at a given time right before using a window of size $w$.
We can think of the algorithm as a random process where
$m$ balls (modelling the messages) are dropped uniformly in
$w$ bins (modelling time slots).  We will show that, if $m\leq w$, for large
enough $m$, with high probability, at least a constant fraction of the
balls fall alone in a bin. Now, we can repeat the process removing
this constant fraction of balls and bins until all balls have fallen
alone. Since nodes do not know $m$, the
outer loop increasing the size of the window is necessary.
The analysis follows.

%\subsection{Analysis}
 
\begin{lemma}
\label{lemma:ballsbins}
For $k\geq m\geq (2{\mathrm e}/(1-{\mathrm e}\delta)^2)(1 + (\beta+1/2)\ln k)$, $0<\delta<1/{\mathrm e}$, $m\leq w$, and $\beta>0$, if
$m$ balls are dropped in $w$ bins uniformly at random, the
probability that the number of bins with exactly one ball is less
than $\delta m$ is at most $1/k^{\beta}$.
\end{lemma}

\begin{proof}
Since a bigger number of bins can only reduce the number of bins with more than one ball, if the claim holds for $w=m$ it also holds for $w>m$. Thus, it is enough to prove the first case.
The probability for a given ball to fall alone in a given bin is $(1/m)(1-1/m)^{m-1} \geq 1/({\mathrm e}m)$.
Let $X_i$ be a random variable that indicates if there is exactly one ball in bin $i$. 
Then, $Pr(X_i=1)\geq 1/{\mathrm e}$.
To handle the dependencies that arise in balls and bins problems, we
approximate the joint distribution of the number of balls in all bins
by assuming the load in each bin is an independent Poisson random
variable with mean 1. Let $X$ be a random variable that indicates the
total number of bins with exactly one ball. Then,
$\mu=E[X]=m/{\mathrm e}$. Using Chernoff-Hoeffding bounds~\cite{book:mitzenmacher},
%\begin{align*}
%  Pr(X\leq\delta m)
%&\leq \exp\left(-\frac{m}{2{\mathrm e}}\left(1-{\mathrm e}\delta\right)^2\right)\textrm{, because $0<\delta<1/{\mathrm e}$}.
%\end{align*}
$Pr(X\leq\delta m)
\leq \exp\left(-m \left(1-{\mathrm e}\delta\right)^2/(2{\mathrm e})\right)$, because $0<\delta<1/{\mathrm e}$.

As shown in~\cite{book:mitzenmacher}, any event that takes place with
probability $p$ in the Poisson case takes place with probability at
most $p{\mathrm e}\sqrt{m}$ in the exact case. Then, we want to show that
%\begin{align*}
%  \exp\left(-\frac{m}{2{\mathrm e}}(1-{\mathrm e}\delta)^2\right) {\mathrm e}\sqrt{m} &\leq k^{-\beta}.
%\end{align*}
 $\exp\left(-m(1-{\mathrm e}\delta)^2/ (2{\mathrm e})\right) {\mathrm e}\sqrt{m} \leq k^{-\beta}$,
which is true for
%\begin{align*}
% m &\geq \frac{2{\mathrm e}}{(1-{\mathrm e}\delta)^2}\left(1+\left(\frac{1}{2}+\beta\right)\ln k\right).
%\end{align*}
$m \geq \frac{2{\mathrm e}}{(1-{\mathrm e}\delta)^2}\left(1+\left(\frac{1}{2}+\beta\right)\ln k\right)$.
\qed
\end{proof}

%%%%%%%%%%%%%%%%%%%%%%%%%%%%%%%%%%%%%%%%%%%%%%%%

\begin{theorem}
For any constant $0 < \delta < 1/{\mathrm e}$, \EBOBO solves \sKS within $4(1+1/\delta)k$ steps with probability at least $1-1/k^c$, for some constant $c>0$ and big enough $k$.
\end{theorem}

\begin{proof}
Consider an execution of the algorithm on $k$ nodes. 
Let a round be the sequence of time steps corresponding to one iteration of the inner loop of Algorithm~\ref{algo}, i.e. the time steps of a window. 
Let a phase be the sequence of rounds corresponding to one iteration of the outer loop of Algorithm~\ref{algo}, i.e. when the window is monotonically reduced.
%Telescoping the running time of each loop, up to a phase where $w\in O(k)$, the claimed time complexity follows. Thus, it is enough to prove that all messages are transmitted whp within that time, which we do as follows.

Consider the first round when $k\leq w< 2k$. 
Assume no message was transmitted successfully before. 
(Any messages transmitted could only reduce the running time.) 
By Lemma~\ref{lemma:ballsbins}, we know that, for $0<\delta<1/\mathrm{e}$ and $\beta>0$, at least $\delta k$ messages are transmitted in this round with probability at least $1-1/k^\beta$, as long as $k\geq \tau$, where $\tau \triangleq (2{\mathrm e}/(1-{\mathrm e}\delta)^2)(1 + (\beta+1/2)\ln k)$. 

%\textcolor{red}{($\beta=2,\delta=1/3 \rightarrow n\geq \tau \rightarrow n \approx 20000$!!!)}

Conditioned on this event, for some $\delta_1\geq \delta$ fraction of messages transmitted in the first round, using the same lemma we know that in the following round at least $\delta(1-\delta_1)k$ messages are transmitted with probability at least $1-1/k^\beta$, as long as $(1-\delta_1)k\geq \tau$. This argument can be repeated for each subsequent round until the number of messages left to be transmitted is less than $\tau$. Furthermore, given that the size of the window is monotonically reduced within a phase until $w=1$, even if the fraction of messages transmitted in each round is just $\delta$, the overall probability of reducing the number of messages left from $k$ to $\tau$ within this phase is at least $(1-1/k^\beta)^{\log_{1/(1-\delta)} (2k)}$.

Consider now the first round of the following phase, i.e. when $2k\leq w< 4k$. Assume that at most $\tau$ nodes still hold a message to be transmitted. Using the union bound, the probability that two or more of $m$ nodes choose a given step in a window of size $w$ is at most $\binom{m}{2}/w^2$. Applying again the union bound, the probability that in any step two or more nodes choose to transmit is at most $\binom{m}{2}/w\leq\binom{\tau}{2}/(2k)=\tau(\tau+1)/(4k)$.

Therefore, using conditional probability, in order to complete the proof, it is enough to show that
\begin{align}
\left(1-\frac{\tau(\tau+1)}{4k}\right) \left(1-\frac{1}{k^\beta}\right)^{\log_{1/(1-\delta)} (2k)} &\geq 1-\frac{1}{k^c},\textrm{ for some constant $c>0$}\nonumber\\
\exp\left(-\frac{\tau(\tau+1)}{4k-\tau(\tau+1)} - \frac{\log_{1/(1-\delta)} (2k)}{k^\beta-1}\right) &\geq \exp\left(-\frac{1}{k^c}\right)\nonumber\\
\frac{\tau(\tau+1)}{4k-\tau(\tau+1)} + \frac{\log_{1/(1-\delta)} (2k)}{k^\beta-1} &\leq \frac{1}{k^c}.\label{bound}
\end{align}

Given that $\delta$ is a constant and fixing $\beta>0$ as a constant, Inequality~\ref{bound} is true for some constant $c<\min\{1,\beta\}$, for big enough $k$.
Telescoping the number of steps up to the first round when $w=4k$, the running time is less than $4k+2k\sum_{i=0}^\infty\sum_{j=0}^\infty (1-\delta)^j/2^i = 4(1+1/\delta)k$.
\qed
\end{proof}

\section{Evaluation}
\label{section:eval}

In order to evaluate the expected behavior of the algorithms  \OFA and \EBOBO, and compare it with the previously proposed
algorithms \LLIBO and \LFA, we have simulated the four algorithms. The simulations measure the number of steps that the algorithms take until the static $k$-selection problem has been solved, i.e., each of the $k$ activated nodes of the Radio Network has delivered its message, for different values of $k$. 
Several of the algorithms have parameters that can be adapted. The value of these parameters is the same for all the simulations of the same algorithm (except the parameter $\varepsilon$ of \LFA that has to depend on $k$). For \EBOBO the parameter is chosen to be $\delta=0.366$. For \OFA the parameter is chosen to be $\delta = 2.72$. For \LFA, the parameters (see their meaning in \cite{FM:kSelJournal}) are chosen to be $\xi_{\delta}=\xi_{\beta}=0.1$ and $\varepsilon \approx 1/(k+1)$, while two values of $\xi_t$ have been used, $\xi_t=1/2$ and $\xi_t=1/10$. Finally, \LLIBO is simulated with parameter $r=2$ (see \cite{bendercontention}).

Figure \ref{f-steps} presents the average number of steps taken by the simulation of the algorithms. The plot shows the the average of 10 runs for each algorithm as a function of $k$. In this figure it can be observed that \LFA takes significantly larger number of steps than the other algorithms for moderately small values of $k$ (up to $10^5$). Beyond $k=10^5$ all algorithms seem to have a similar behavior.

\begin{figure}[t]
\centering\includegraphics[width=0.6\textwidth]{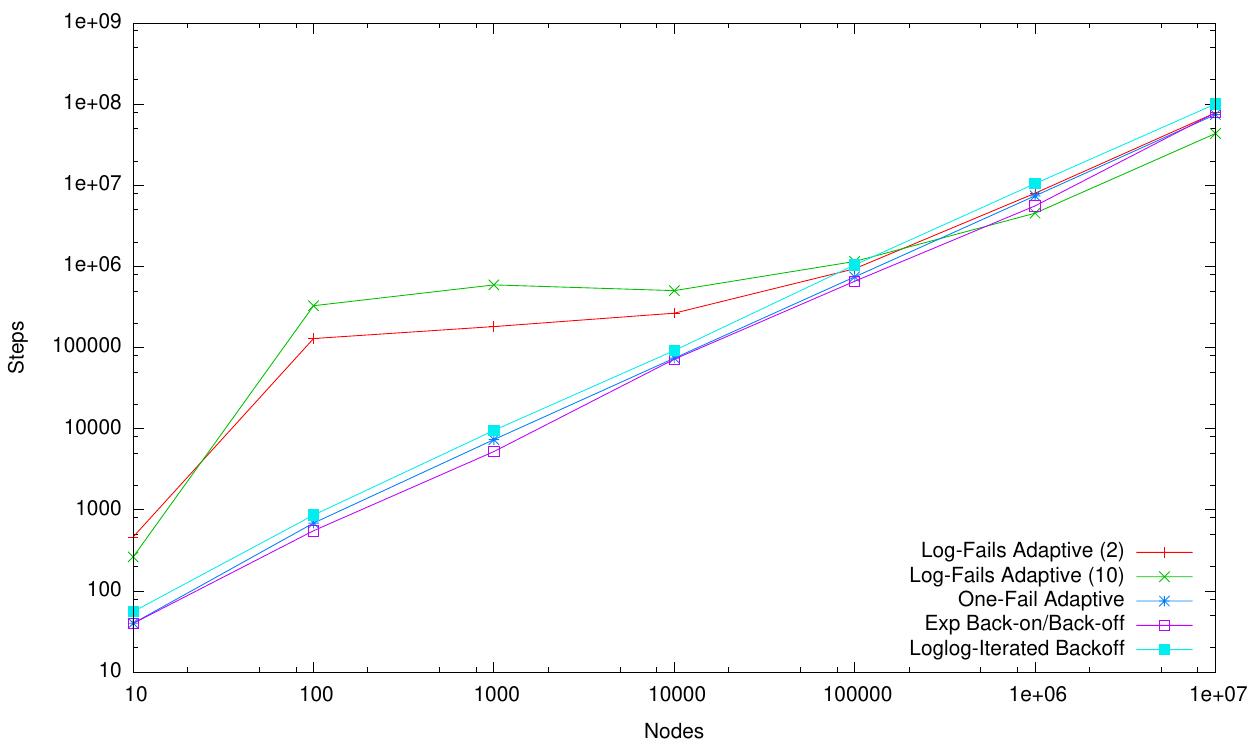}
\caption{Number of steps to solve static $k$-selection, per number of nodes $k$.}
\label{f-steps}
\end{figure}

\begin{table}[tdp]
\begin{center}
\begin{tabular}{|l|r|r|r|r|r|r|r|r|} \hline
$k$ & $10$ & $10^2$ & $10^3$ & $10^4$ & $10^5$ & $10^6$ & $10^7$ & Analysis\\ \hline \hline
\LFA $\xi_t=1/2$ & 46.4  & 1292.4 & 181.9 & 26.6 & 9.4 & 8.0 & 7.8 & 7.8 \\ \hline
\LFA $\xi_t=1/10$ & 26.3 & 3289.2 & 593.8 & 50.3 & 11.5 & 4.5 & 4.4 & 4.4\\ \hline
\OFA & 4.0 & 6.9 & 7.4 & 7.4 & 7.4 & 7.4 & 7.4 & 7.4 \\ \hline
\EBOBO & 4.0 & 5.5 & 5.2 & 7.2 & 6.6 & 5.6 & 7.9 & 14.9\\ \hline
\LLIBO & 5.6 & 8.6 & 9.6 & 9.2 & 10.5 & 10.5 & 10.1 & $\Theta \left( \frac{\log \log k}{\log \log \log k} \right)$ \\ \hline
\end{tabular}
\end{center}
\caption{Ratio steps/nodes as a function of the number of nodes $k$.}
\label{t-ratio-nodes}
\end{table}%

A higher level of detail can be obtained by observing Table \ref{t-ratio-nodes}, which presents the ratio obtained by dividing the number of steps
(plotted in Figure \ref{f-steps}) by the value of $k$, for each $k$ and each algorithm. In this table, the bad behavior of \LFA for moderate values of $k$
can be observed, with values of the ratio well above those for large $k$. It seems like the value of $\xi_t$ used has an impact in this ratio, so that
the smaller value $\xi_t=1/10$ causes larger ratio values. Surprisingly, for large values of $k$ ($k \geq 10^6$), the ratios observed are almost exactly the
constant factors of $k$ obtained from the analysis \cite{FM:kSelJournal}. (Recall that all the analyses we refer to are with high probability while the simulation results are averages.) This may indicate that the analysis with high probability is very tight and that the term $O(\log^2(1/\varepsilon))$ that appears in the complexity expression is mainly relevant for moderate values of $k$. The ratio obtained for large $k$ by \LFA with $\xi_t=1/10$ is the smallest we have obtained in the
set of simulations.
\LLIBO, on its hand, seems to have a constant ratio of around $10$. In reality
this ratio is not constant but, since it is sublogarithmic, this fact can not be observed for the (relatively small) values of $k$ simulated.

Regarding the ratios obtained for the algorithms proposed in this paper, they seem to show that the
constants obtained in the analyses (with high probability) are very accurate.
Starting at moderately large values of $k$ ($10^3$ and up) the ratio for \OFA becomes very stable and
equal to the value of $7.4$ obtained in the analysis.
The ratios for the \EBOBO simulations, on their hand, move between 4 and 8, while the analysis for the value
of $\delta$ used yields a constant factor of 14.9. Hence, the ratios are off by only a small constant factor. 
To appreciate these values it is worth to note that the smallest ratio
expected by any algorithm in which nodes use the same probability at any step is $e$, so these values are only a small
factor away from this optimum ratio.
In summary, the algorithms proposed here have small and stable ratios for all values of $k$ considered.
\section{Conclusions and Open Problems}
\label{section:conclude}

In this work, we have shown optimal randomized protocols (up to constants) for \sKS in \RNs that do not require any knowledge on the number of contenders. 
Future work includes the study of the dynamic version of the problem when messages arrive at different times under the same model, either assuming statistical or adversarial arrivals. The stability of monotonic strategies (exponential back-off) has been studied in~\cite{bendercontention}. In light of the improvements obtained for batched arrivals, the application of non-monotonic strategies to the dynamic problem is promising.

%\newpage

\bibliographystyle{abbrv}
\bibliography{./cocoon114}

\newpage
\appendix
\section*{Appendix}
\section{Lemmata of the analysis of \OFA }

%--------------------------------------------------------------------------------------------
%\subsection{Analysis}
%\label{section:anal}

%The analysis structure of Algorithm~\ref{alg} presented in this paper resembles the techniques used in~\cite{FM:kSelJournal}. However, due to the differences between both algorithms, there are a large number of details that have to be particularly suited for this case. 
For clarity, Algorithms AT and BT are analyzed separately taking into account in both the presence of the other. Communication steps are referred to by the name of the algorithm used, i.e. a communication step is either an AT-step or a BT-step.
The following notation will be used throughout the analysis. 

Let $\kappa$ be the number of messages not delivered yet (i.e., the number of active nodes), called the \emph{density}, and let $\widetilde{\kappa}$ be called the \emph{density estimator}.
Consider the execution of Algorithm~\ref{alg} divided in \emph{rounds} as follows.
The first round begins with the first step of the execution, and a new round starts on each step that $\widetilde{\kappa}$ reaches or exceeds a multiple of $\tau\triangleq 300\delta \ln (1+k)$ for the first time. (Hence, a new round may start only in an AT-step.)
More precisely, let the rounds be numbered as $r\in\{1,2,\dots\}$ and the AT-steps within a round as $t\in\{1,2,\dots\}$.
Let $T_r$ be the set of AT-steps of round $r$.
Let $\widetilde{\kappa}_{r,t}$ be the density estimator used at the AT-step $t$ of round $r$. 
Then, 
$$\forall i,j,t \in\mathbb{N} : \widetilde{\kappa}_{j,1}\geq (j-1)\tau \land ( (i<j \land t \in T_i) \Rightarrow \widetilde{\kappa}_{i,t}< (j-1)\tau).$$
Thus, round $1$ is the sequence of AT-steps from initialization when $\widetilde{\kappa}=1$ until the last step before $\widetilde{\kappa}\geq\tau$ for the first time, round $2$ begins on the AT-step when $\widetilde{\kappa}\geq\tau$ for the first time and ends right before $\widetilde{\kappa}\geq 2\tau$ for the first time, and so on.
Let $X_{r,t}$ be an indicator random variable such that, $X_{r,t} = 1$ if a message is delivered at the AT-step $t$ of round $r$, and $X_{r,t}=0$ otherwise.
Let $\kappa_{r,t}$ be the density at the beginning of the AT-step $t$ of round $r$.
Then, $Pr(X_{r,t}=1)= (\kappa_{r,t}/\widetilde{\kappa}_{r,t}) (1-1/\widetilde{\kappa}_{r,t})^{\kappa_{r,t}-1}$ is the probability of a successful transmission in the AT-step $t$ of round $r$.
Also, for a round $r$, let the number of messages delivered in the interval of AT-steps $[1,t)$ of $r$ be denoted as $\sigma_{r,t}$.

The following intermediate results will be useful.
First, we state the following useful facts.
\begin{fact}{\cite[\S 2.68]{book:mitrinovic}}
\label{fact}
$e^{x/(1+x)} \leq 1+x \leq e^{x}, 0<|x|<1$.
\end{fact}
\begin{fact}
\label{claim}
Given any constant $a>1$, the function $f:\mathbb{R}^+\to\mathbb{R}^+$, such that $f(x)\triangleq(a/x)(1-1/x)^{a-1}$, is non decreasing for $x<a$ and maximized for $x=a$.
\end{fact}

%--------------------------------------------------------------------------------------LEMMA POS CORR

\begin{lemma}
\label{lemma:poscorr}
For any round $r$ 
and any $t,t+1\in T_r$ such that $\widetilde{\kappa}_{r,t}<\kappa_{r,t}$,
if $\widetilde{\kappa}_{r,t+1} = \widetilde{\kappa}_{r,t}+1$,
then $Pr(X_{r,t}=1)\leq Pr(X_{r,t+1}=1)$.
\end{lemma}
\begin{proof}
We want to show
\begin{align*}
\frac{\kappa_{r,t}}{\widetilde{\kappa}_{r,t}} \left(1-\frac{1}{\widetilde{\kappa}_{r,t}}\right)^{\kappa_{r,t}-1} 
&\leq \frac{\kappa_{r,t+1}}{\widetilde{\kappa}_{r,t+1}} \left(1-\frac{1}{\widetilde{\kappa}_{r,t+1}}\right)^{\kappa_{r,t+1}-1}
\end{align*}

Given that the density estimator was increased from $t$ to $t+1$ and that $\delta>1$, we know that there was no successful transmission, neither at the AT-step $t$, nor at the BT-step between the AT-steps $t$ and $t+1$ (see Algorithm~\ref{alg}). Thus, $\kappa_{r,t+1} = \kappa_{r,t}$. Replacing,
\begin{align*}
\frac{\kappa_{r,t}}{\widetilde{\kappa}_{r,t}} \left(1-\frac{1}{\widetilde{\kappa}_{r,t}}\right)^{\kappa_{r,t}-1} 
&\leq \frac{\kappa_{r,t}}{\widetilde{\kappa}_{r,t}+1} \left(1-\frac{1}{\widetilde{\kappa}_{r,t}+1}\right)^{\kappa_{r,t}-1}
\end{align*}

Which, due to Fact~\ref{claim}, is true for $\widetilde{\kappa}_{r,t} < \kappa_{r,t}$.
\qed
\end{proof}

%--------------------------------------------------------------------------------------LEMMA NEG CORR

\begin{lemma}
\label{lemma:negcorr}
For any round $r$ 
where $\widetilde{\kappa}_{r,1} \leq \kappa_{r,1}-\gamma$, $\gamma\geq (\delta-1)(3-\delta)/(\delta-2)\geq 0$,
and any $t,t+1\in T_r$ such that
$\delta<\widetilde{\kappa}_{r,t} \leq \kappa_{r,t}$,
and $\delta-1 <(\kappa_{r,t}-\gamma)(\kappa_{r,t}-\gamma-1)/(\kappa_{r,t}-\gamma+1)$,
if $\widetilde{\kappa}_{r,t+1}<\widetilde{\kappa}_{r,t}$,
then $Pr(X_{r,t}=1)\geq Pr(X_{r,t+1}=1)$.
\end{lemma}

\begin{proof}
We want to show
\begin{align*}
\frac{\kappa_{r,t}}{\widetilde{\kappa}_{r,t}} \left(1-\frac{1}{\widetilde{\kappa}_{r,t}}\right)^{\kappa_{r,t}-1} 
&\geq \frac{\kappa_{r,t+1}}{\widetilde{\kappa}_{r,t+1}} \left(1-\frac{1}{\widetilde{\kappa}_{r,t+1}}\right)^{\kappa_{r,t+1}-1}.
\end{align*}

Given that the density estimator was reduced from $t$ to $t+1$, we know that, either at the AT-step $t$, or at the BT-step between the AT-steps $t$ and $t+1$, or in both, there were successful transmissions (see Algorithm~\ref{alg}). Thus, we have to show that
%\begin{align*}
%\frac{\kappa_{r,t}}{\widetilde{\kappa}_{r,t}} 
%&\left(1-\frac{1}{\widetilde{\kappa}_{r,t}}\right)^{\kappa_{r,t}-1} 
%\geq \\
%&\left\{ \begin{array}{ll} 
%\frac{\kappa_{r,t}-1}{\widetilde{\kappa}_{r,t}-\delta} \left(1-\frac{1}{\widetilde{\kappa}_{r,t}-\delta}\right)^{\kappa_{r,t}-2} 
%&\textrm{ if BT-step not successful,}\\
%\frac{\kappa_{r,t}-1}{\widetilde{\kappa}_{r,t}-\delta+1} \left(1-\frac{1}{\widetilde{\kappa}_{r,t}-\delta+1}\right)^{\kappa_{r,t}-2}
%&\textrm{ if AT-step not successful,}\\
%\frac{\kappa_{r,t}-2}{\widetilde{\kappa}_{r,t}-2\delta} \left(1-\frac{1}{\widetilde{\kappa}_{r,t}-2\delta}\right)^{\kappa_{r,t}-3}
%&\textrm{ if both steps successful.}
%\end{array} \right. 
%\end{align*}
%
\begin{align}
\frac{\kappa_{r,t}}{\widetilde{\kappa}_{r,t}} \left(1-\frac{1}{\widetilde{\kappa}_{r,t}}\right)^{\kappa_{r,t}-1} &\geq \label{lhs}
\end{align}
\begin{align}
\frac{\kappa_{r,t}-1}{\widetilde{\kappa}_{r,t}-\delta+1} \left(1-\frac{1}{\widetilde{\kappa}_{r,t}-\delta+1}\right)^{\kappa_{r,t}-2}
\textrm{, if AT-step not successful,}\label{2ndeq}\\
\frac{\kappa_{r,t}-1}{\widetilde{\kappa}_{r,t}-\delta} \left(1-\frac{1}{\widetilde{\kappa}_{r,t}-\delta}\right)^{\kappa_{r,t}-2} 
\textrm{, if BT-step not successful,}\label{1steq}\\
\frac{\kappa_{r,t}-2}{\widetilde{\kappa}_{r,t}-2\delta} \left(1-\frac{1}{\widetilde{\kappa}_{r,t}-2\delta}\right)^{\kappa_{r,t}-3}
\textrm{, if both steps successful.}\label{3rdeq}
\end{align}

%That (\ref{lhs}) $\geq$ (\ref{2ndeq}) was proved to be true in the proof of Lemma 1 in~\cite{FM:kSelJournal} for the conditions of this lemma.

%Given that 
%$\kappa_{r,t}-1 \geq
%\widetilde{\kappa}_{r,t}-\delta+1 \geq
%\widetilde{\kappa}_{r,t}-\delta$,
%we know from Fact~\ref{claim} that (\ref{2ndeq}) $\geq$ (\ref{1steq}). 
%Then, transitively, we know that (\ref{lhs}) $\geq$ (\ref{1steq}) for the conditions of this lemma.

%With a change of variables, (\ref{lhs}) $\geq$ (\ref{1steq}) implies (\ref{1steq}) $\geq$ (\ref{3rdeq}) as long as \textcolor{red}{????}
%Then, transitively, it is also (\ref{lhs}) $\geq$ (\ref{3rdeq}).

That (\ref{lhs}) $\geq$ (\ref{2ndeq}) was proved to be true in the proof of Lemma 3.2 in~\cite{FM:kSelJournal} (Eq. (3.1) in the proof),
for the conditions of this lemma.
Given that 
$\kappa_{r,t}-1 \geq
\widetilde{\kappa}_{r,t}-\delta+1 \geq
\widetilde{\kappa}_{r,t}-\delta$,
we know from Fact~\ref{claim} that (\ref{2ndeq}) $\geq$ (\ref{1steq}). 
Then, transitively, we know that (\ref{lhs}) $\geq$ (\ref{1steq}) for the conditions of this lemma.

In the proof of Lemma 3.2 in~\cite{FM:kSelJournal} (Eq. (3.2) in the proof) was also proved (with a change of variable to meet
the conditions of this lemma) that (\ref{lhs}) is at least as large as
\begin{equation}
\frac{\kappa_{r,t}-2}{\widetilde{\kappa}_{r,t}-2\delta+2} \left(1-\frac{1}{\widetilde{\kappa}_{r,t}-2\delta+2}\right)^{\kappa_{r,t}-3}
\textrm. \label{4theq}
\end{equation}
Given that 
$\kappa_{r,t}-2 \geq
\widetilde{\kappa}_{r,t}-2\delta+2 \geq
\widetilde{\kappa}_{r,t}-2\delta$,
we know from Fact~\ref{claim} that (\ref{4theq}) $\geq$ (\ref{3rdeq}). 
Then, transitively, we know that (\ref{lhs}) $\geq$ (\ref{3rdeq}) for the conditions of this lemma.
%
%Given that 
%$(\kappa_{r,t}-1)/(\widetilde{\kappa}_{r,t}-\delta) \geq
%(\kappa_{r,t}-1)/(\widetilde{\kappa}_{r,t}-\delta+1) >1$,
%and that 
%$(\kappa_{r,t}-2)/(\widetilde{\kappa}_{r,t}-2\delta) \geq
%(\kappa_{r,t}-1)/(\widetilde{\kappa}_{r,t}-\delta+1) >1$ 
%because
%$\widetilde{\kappa}_{r,t}\geq 2/\delta$, 
%due to Fact~\ref{claim}, it is enough to prove
%\begin{align*}
%\frac{\kappa_{r,t}}{\widetilde{\kappa}_{r,t}} \left(1-\frac{1}{\widetilde{\kappa}_{r,t}}\right)^{\kappa_{r,t}-1} 
%&\geq \frac{\kappa_{r,t}-1}{\widetilde{\kappa}_{r,t}-\delta+1} \left(1-\frac{1}{\widetilde{\kappa}_{r,t}-\delta+1}\right)^{\kappa_{r,t}-2}.
%\end{align*}
%
%The latter inequality was proved to be true in the proof of Lemma 1 in~\cite{FM:kSelJournal} for the conditions of this lemma.
\qed
\end{proof}

%-------------------------------------------------------------------------------------------LEMMA PROB LB

\begin{lemma}
\label{lemma:problb}
For any $\beta$ such that $(\delta+1)\ln\beta>1$,
and for any round $r$ where $\kappa_{r,1}-\alpha \leq \widetilde{\kappa}_{r,1} < \kappa_{r,1}$, $\alpha\geq 0$ and for any AT-step $t$ in $r$ such that 
%\begin{align*}
%1<\widetilde{\kappa}_{r,t} &\leq \kappa_{r,t}\\
%\sigma_{r,t} &\leq \kappa_{r,1}\frac{\ln\beta-1}{(\delta+1)\ln\beta-1}-\frac{(\alpha+1-t)\ln\beta-1}{(\delta+1)\ln\beta-1}.
%\end{align*}
$1<\widetilde{\kappa}_{r,t} \leq \kappa_{r,t}$ and $\sigma_{r,t} \leq \kappa_{r,1}\frac{\ln\beta-1}{(\delta+1)\ln\beta-1}-\frac{(\alpha+1-t)\ln\beta-1}{(\delta+1)\ln\beta-1}$,
the probability of a successful transmission is at least $Pr(X_{r,t}=1)\geq 1/\beta$.
\end{lemma}

\begin{proof}
We want to show
%\begin{align*}
%\frac{\kappa_{r,t}}{\widetilde{\kappa}_{r,t}}\left(1-\frac{1}{\widetilde{\kappa}_{r,t}}\right)^{\kappa_{r,t}-1}
%&\geq 1/\beta.
%\end{align*}
$\frac{\kappa_{r,t}}{\widetilde{\kappa}_{r,t}}\left(1-\frac{1}{\widetilde{\kappa}_{r,t}}\right)^{\kappa_{r,t}-1} \geq 1/\beta$.
Because $\widetilde{\kappa}_{r,t} \leq \kappa_{r,t}$ it is enough to prove
%\begin{align*}
%\left(1-\frac{1}{\widetilde{\kappa}_{r,t}}\right)^{\kappa_{r,t}-1}
%&\geq 1/\beta.
%\end{align*}
$\left(1-\frac{1}{\widetilde{\kappa}_{r,t}}\right)^{\kappa_{r,t}-1} \geq 1/\beta$.
Because $\widetilde{\kappa}_{r,t}>1$, using Fact~\ref{fact},
%\begin{align*}
%\exp\left(-\frac{\kappa_{r,t}-1}{\widetilde{\kappa}_{r,t}-1}\right)
%&\geq 1/\beta\\
%\kappa_{r,t}-1 &\leq (\widetilde{\kappa}_{r,t}-1)\ln\beta.
%\end{align*}
we obtain $\exp\left(-\frac{\kappa_{r,t}-1}{\widetilde{\kappa}_{r,t}-1}\right) \geq 1/\beta$, which holds if
$\kappa_{r,t}-1\leq (\widetilde{\kappa}_{r,t}-1)\ln\beta$.

Given that nodes are active until their message is delivered, we know that $\kappa_{r,t}=\kappa_{r,1}-\sigma_{r,t}$.
Additionally, we know that $\widetilde{\kappa}_{r,t}=\widetilde{\kappa}_{r,1}-\delta\sigma_{r,t}+t-\sigma_{r,t}$ (see Algorithm~\ref{alg}) and that $\widetilde{\kappa}_{r,1}\geq\kappa_{r,1}-\alpha$ by hypothesis. 
Replacing, we obtain
%\begin{align*}
%\kappa_{r,1}-\sigma_{r,t}-1 &\leq \ln\beta(\kappa_{r,1}-\alpha-(\delta+1)\sigma_{r,t}+t-1)\\
%\sigma_{r,t} &\leq \kappa_{r,1}\frac{\ln\beta-1}{(\delta+1)\ln\beta-1}-\frac{(\alpha+1-t)\ln\beta-1}{(\delta+1)\ln\beta-1},
%\end{align*}
$\kappa_{r,1}-\sigma_{r,t}-1 \leq \ln\beta(\kappa_{r,1}-\alpha-(\delta+1)\sigma_{r,t}+t-1)$. This holds if 
$\sigma_{r,t} \leq \kappa_{r,1}\frac{\ln\beta-1}{(\delta+1)\ln\beta-1}-\frac{(\alpha+1-t)\ln\beta-1}{(\delta+1)\ln\beta-1}$,
for any $\beta$ such that $(\delta+1)\ln\beta>1$.
\qed
\end{proof}

%--------------------------------------------------------------------------------------------------------------------LEMMA AT ALG

The following lemma, shows the efficiency and correctness of the AT-algorithm. 
\begin{lemma}
\label{lemma:dense}
For any $e<\delta\leq\sum_{j=1}^{5}(5/6)^{j}$,
if the number of messages to deliver is more than 
\begin{align*}
M &= \frac{(\delta+1)\ln\delta-1}{\ln\delta-1} S +\frac{(\gamma+2\tau+1)\ln\delta-1}{\ln\delta-1},
%S &= 2 \sum_{j=0}^{4}(5/6)^{j}\tau \\
%\gamma &= \frac{(\delta-1)(3-\delta)}{\delta-2}
\end{align*}
where $S = 2 \sum_{j=0}^{4}(5/6)^{j}\tau$ and  $\gamma = (\delta-1)(3-\delta)/(\delta-2)$,
after running the AT-algorithm for $(\delta+1)k$ AT-steps, 
the number of messages left to deliver is reduced to at most $M$ with probability at least $1-1/(1+k)$.
\end{lemma}

\begin{proof}
Consider the first round $r$ such that 
$\kappa_{r,1}-\gamma-2\tau \leq \widetilde{\kappa}_{r,1} < \kappa_{r,1}-\gamma-\tau$.
%\begin{align}
%\kappa_{r,1}-\gamma-2\tau \leq \widetilde{\kappa}_{r,1} < \kappa_{r,1}-\gamma-\tau. \label{rcond}
%\end{align}
Unless the number of messages left to deliver is reduced to at most $M$ before, such a round exists because the density estimator is increased only one by one (see Algorithm~\ref{alg}).
Furthermore, given 
%the upper bound on $\widetilde{\kappa}_{r,1}$
that $\widetilde{\kappa}_{r,1} < \kappa_{r,1}-\gamma-\tau$,
even if no message is transmitted in round $r$, it holds that $\widetilde{\kappa}_{r,t} < \kappa_{r,t}$ for any $t$ in $r$ by the definition of a round.
Additionally, we will show that, before leaving round $r$, at least $\tau$ messages are delivered with big enough probability so that in some future round $r''>r$ the condition $\kappa_{r'',1}-\gamma-2\tau \leq \widetilde{\kappa}_{r'',1} < \kappa_{r'',1}-\gamma-\tau$ holds again. 

Consider round $r$ divided in consecutive sub-rounds of size $\tau, 5/6 \tau, (5/6)^2 \tau, \dots$ (The fact that a number of steps is an integer is omitted throughout for clarity.) More specifically, the sub-round $S_1$ is the set of AT-steps in the interval $(0,\tau]$ and, for $i\geq 2$, the sub-round $S_i$ is the set of steps in the interval $((5/6)^{i-2}\tau,(5/6)^{i-1}\tau]$.
Thus, denoting $|S_i|=\tau_i$ for all $i\geq 1$, it is $\tau_1=\tau$ and $\tau_{i}=(5/6)\tau_{i-1}$ for $i\geq 2$.
For each $i\geq 1$, let $Y_i$ be a random variable such that $Y_i=\sum_{t\in S_i} X_{r,t}$.
Even if no message is delivered, round $r$ still has at least the sub-round $S_1$ by the definition of a round.
Given that each message delivered delays the end of round $r$ in $\delta$ AT-steps (see Algorithm~\ref{alg}), for $i\geq 2$, the existence of sub-round $S_i$ is conditioned on $Y_{i-1}\geq 5\tau_{i-1}/(6\delta)$.
We show that with big enough probability round $r$ has $5$ sub-rounds and at least $\tau$ messages are delivered as follows.

Even if messages are delivered in every step of the $5$ sub-rounds, including messages delivered in BT-steps,
given that $\kappa_{r,1} >M$,  the total number of messages delivered is
%\begin{align*}
%S &= M \frac{\ln\delta-1}{(\delta+1)\ln\delta-1} - \frac{(\gamma+2\tau+1)\ln\delta-1}{(\delta+1)\ln\delta-1}\\
%&< \kappa_{r,1} \frac{\ln\delta-1}{(\delta+1)\ln\delta-1} - \frac{(\gamma+2\tau+1)\ln\delta-1}{(\delta+1)\ln\delta-1}\\
%&< \kappa_{r,1} \frac{\ln\delta-1}{(\delta+1)\ln\delta-1} - \frac{(\gamma+2\tau+1-t)\ln\delta-1}{(\delta+1)\ln\delta-1}.
%\end{align*}
\begin{align*}
S 
&= M \frac{\ln\delta-1}{(\delta+1)\ln\delta-1} - \frac{(\gamma+2\tau+1)\ln\delta-1}{(\delta+1)\ln\delta-1}\\
&< \kappa_{r,1} \frac{\ln\delta-1}{(\delta+1)\ln\delta-1} - \frac{(\gamma+2\tau+1-t)\ln\delta-1}{(\delta+1)\ln\delta-1}.
\end{align*}
Taking $\alpha=\gamma+2\tau$ and $\beta=\delta$, Lemma~\ref{lemma:problb} can be applied because $\kappa_{r,1}-\alpha \leq \widetilde{\kappa}_{r,1} < \kappa_{r,1}$ and $(\delta+1)\ln\beta>1$. Hence, the expected number of messages delivered in $S_i$ is $E[Y_i]\geq \tau_i/\delta$. 

In order to use Lemmas~\ref{lemma:poscorr} and~\ref{lemma:negcorr}, we verify first their preconditions.
As argued above, $\widetilde{\kappa}_{r,t} < \kappa_{r,t}$ during the whole round. 
Thus, Lemma~\ref{lemma:poscorr} can be applied.
As for Lemma~\ref{lemma:negcorr}, we know that 
$\delta < \widetilde{\kappa}_{r,t}$ (see Algorithm~\ref{alg}),
$\widetilde{\kappa}_{r,1} \leq \kappa_{r,1}-\gamma$ in the round under consideration, and
$\gamma \geq (\delta-1)(3-\delta)/(\delta-2)$ by hypothesis. 
%Antonio: Revisar la siguiente frase
Then, $(\kappa_{r,t}-\gamma)(\kappa_{r,t}-\gamma-1)/(\kappa_{r,t}-\gamma+1)>\delta-1$ follows
from $M \geq 2\delta+\gamma-1$ and $\kappa_{r,t} > M$.
%\begin{align*}
%(\kappa_{r,t}-\gamma)(\kappa_{r,t}-\gamma-1)/(\kappa_{r,t}-\gamma+1)&>\delta-1 \\
%(\kappa_{r,t}-\gamma)^2 &> \delta(\kappa_{r,t}-\gamma)+(\delta-1) \\
%(\kappa_{r,t}-\gamma)^2 &> (2\delta-1)(\kappa_{r,t}-\gamma) \\
%\kappa_{r,t} &> 2\delta+\gamma-1\\
%%M &\geq 2\delta+\gamma-1\\
%M &\geq 2\delta+\gamma-1.
%\end{align*}

Then, we use Lemmas~\ref{lemma:poscorr} and~\ref{lemma:negcorr} as follows.
If $\widetilde{\kappa}_{r,t+1} < \widetilde{\kappa}_{r,t}+1$, Lemma~\ref{lemma:negcorr} holds and $Pr(X_{r,t+1}=1)\leq Pr(X_{r,t}=1)$.
On the other hand, 
if $\widetilde{\kappa}_{r,t+1} = \widetilde{\kappa}_{r,t}+1$, Lemma~\ref{lemma:poscorr} holds and $Pr(X_{r,t+1}=1) > Pr(X_{r,t}=1)$.
Assuming instead that $Pr(X_{r,t+1}=1)=Pr(X_{r,t}=1)$ can not increase the value of $Y_i$.
Therefore, in order to bound from below $Y_i$, we assume that the variables $X_{r,t},X_{r,t+1}$ for any $t$ in $r$ are not positively correlated and we use the following Chernoff-Hoeffding bound~\cite{book:mitzenmacher}.

For $0<\varphi<1$,
\begin{displaymath} 
\left\{ \begin{array}{l} 
Pr(Y_1 \leq (1-\varphi) \tau_1/\delta) \leq e^{-\varphi^2\tau_1/(2\delta)}\\
Pr(Y_i \leq (1-\varphi) \tau_i/\delta|Y_{i-1}\geq 5\tau_{i-1}/(6\delta)) \leq 
e^{-\varphi^2\tau_i/(2\delta)}, \forall i: 2\leq i\leq 5.
\end{array} 
\right. 
\end{displaymath} 

Taking $\varphi=1/6$,
\begin{displaymath} 
\left\{ \begin{array}{l} 
Pr(Y_1 \leq 5\tau_1/(6\delta)) \leq e^{-\varphi^2 300\ln(1+k)/2}\\
Pr(Y_i \leq 5\tau_i/(6\delta)|Y_{i-1}\geq 5\tau_{i-1}/(6\delta)) \leq
e^{-\varphi^2(5/6)^{i-1}300\ln(1+k)/2},\\ \forall i: 2\leq i\leq 5.
\end{array} \right. 
\end{displaymath} 

\begin{displaymath} 
\left\{ \begin{array}{l} 
Pr(Y_1 \leq 5\tau_1/(6\delta)) < e^{-2\ln (1+k)}\\
Pr(Y_i \leq 5\tau_i/(6\delta)|Y_{i-1}\geq 5\tau_{i-1}/(6\delta)) <
e^{-2\ln (1+k)}, \forall i: 2\leq i\leq 5.
\end{array} \right. 
\end{displaymath} 

Given that $e^{-2\ln (1+k)} \leq 1/(1+k(1+k))$, more than $(5/(6\delta))\tau_i$ messages are delivered in any sub-round $S_i$ with probability at least $1-1/(1+k(1+k))$. 
Given that each success delays the end of round $r$ in $\delta$ AT-steps, we know that, for $1\leq i\leq 4$, sub-round $S_{i+1}$ exists with probability at least $1-1/(1+k(1+k))$.
If, after any sub-round, the number of messages left to deliver is at most $M$, we are done. 
Otherwise, conditioned on these events, the total number of messages delivered over the $5$ sub-rounds is at least
$\sum_{j=1}^{5}Y_j>\sum_{j=1}^{5}(5/(6\delta))^j \delta^{j-1}\tau = (\tau/\delta) \sum_{j=1}^{5}(5/6)^j \geq \tau$ because $\delta\leq\sum_{j=1}^{5}(5/6)^{j}$.

Thus, the same analysis can be repeated over the next round $r''$ such that $\kappa_{r'',1}-\gamma-\tau \leq \widetilde{\kappa}_{r'',1} < \kappa_{r'',1}-\gamma$. 
Unless the number of messages left to deliver is reduced to at most $M$ before, such a round $r''$ exists by the same argument used to prove the existence of round $r$.
The same analysis is repeated over various rounds until all messages have been delivered or the number of messages left is at most $M$. Then, using conditional probability, the overall probability of success is at least $(1-1/(1+k(1+k)))^k$. Using Fact~\ref{fact} twice, that probability is at least $1-1/(1+k)$.

It remains to be shown the time complexity of the AT algorithm. 
The difference between the number of messages to deliver and the density estimator right after initialization is less than $k$ (see Algorithm~\ref{alg}).
This difference is increased with each message delivered by at most $\delta$.
Then, that difference is never more than $k(\delta+1)$.
Given that the density estimator never exceeds the actual density, the claim follows.
\qed
\end{proof}

The following lemma shows the correctness and time complexity of the BT Algorithm.

\begin{lemma}
\label{lemma:sparse}
If the number of messages left to deliver is at most 
\begin{align*}
M &= \frac{(\delta+1)\ln\delta-1}{\ln\delta-1} S +\frac{(\gamma+2\tau+1)\ln\delta-1}{\ln\delta-1},
%S &= 2 \sum_{j=0}^{4}(5/6)^{j}\tau \\
%\gamma &= \frac{(\delta-1)(3-\delta)}{\delta-2},
\end{align*}
where $S = 2 \sum_{j=0}^{4}(5/6)^{j}\tau$ and $\gamma = (\delta-1)(3-\delta)/(\delta-2)$,
there exists a constant $\xi>0$ such that, 
after running the BT Algorithm for $\xi\log k\ln(1+k)$ BT-steps, all messages are delivered with probability at least $1-1/(1+k)$.
\end{lemma}

\begin{proof}
Let $\sigma(t)$ be the number of messages delivered up to BT-step $t$.
Then, the probability that a given message is not delivered at BT-step $t$ is
\begin{align*}
1-\frac{1}{1+\log(\sigma(t)+1)}\left(1-\frac{1}{1+\log(\sigma(t)+1)}\right)^{k-\sigma(t)-1}.
\end{align*}

Which, given that $\sigma(t)\geq k-M$, is at most
\begin{align*}
1-\frac{1}{1+\log(k+1)}\left(1-\frac{1}{1+\log(k-M+1)}\right)^{M-1}.
\end{align*}

Therefore, the probability that a given message is not delivered for $\xi\log k\ln(1+k)$ BT-steps is at most
\begin{align*}
\left(1-\frac{1}{1+\log(k+1)}\left(1-\frac{1}{1+\log(k-M+1)}\right)^{M-1}\right)^{\xi\log k\ln(1+k)}. 
\end{align*}

%Which, using Fact~\ref{fact}, is at most $1/(1+k)$, for some constant $\xi>0$.
Thus, we want to show, 
\begin{multline*}
\left(1-\frac{1}{1+\log(k+1)}\left(1-\frac{1}{1+\log(k-M+1)}\right)^{M-1}\right)^{\xi\log k\ln(1+k)} \\ \leq 1/(1+k).
\end{multline*}

Using Fact~\ref{fact} twice,
\begin{align*}
%\left(1-\frac{1}{1+\log(k+1)}\exp\left(-\frac{M-1}{\log(k-M+1)}\right)\right)^{\xi\log k\ln(1+k)} &\leq 1/(1+k)\\ 
%\exp\left(-\frac{\xi\log k\ln(1+k)}{1+\log(k+1)}\exp\left(-\frac{M-1}{\log(k-M+1)}\right)\right) &\leq 1/(1+k)\\ 
%\exp\left(\frac{\xi\log k\ln(1+k)}{1+\log(k+1)}\exp\left(-\frac{M-1}{\log(k-M+1)}\right)\right) &\geq 1+k\\ 
%\frac{\xi\log k\ln(1+k)}{1+\log(k+1)}\exp\left(-\frac{M-1}{\log(k-M+1)}\right) &\geq \ln(1+k)\\ 
\xi &\geq \frac{1+\log(k+1)}{\log k}\exp\left(\frac{M-1}{\log(k-M+1)}\right)
\end{align*}

Since $M=c\ln(1+k)$, for some constant $c$,
\begin{align*}
\xi &\geq \frac{1+\log(k+1)}{\log k}\exp\left(\frac{c\ln(1+k)-1}{\log(k-c\ln(1+k)+1)}\right)
\end{align*}

Which is at most a constant.
\qed
\end{proof}

%%-------------------------------------------------------------------------------------------------------------THEOREM
%The following theorem, which is direct consequence of Lemmas~\ref{lemma:dense} and~\ref{lemma:sparse} and the fact that both algorithms are interleaved, establishes the main result.
%\begin{theorem}
%For any $e<\delta\leq\sum_{j=1}^{5}(5/6)^j$ and for any one-hop \RN under the model detailed in Section~\ref{section:intro}, \OFA solves \sKS within $2(\delta+1)k+O(\log^2 k)$ communication steps, with probability at least $1-2/(1+k)$.
%\end{theorem}

\end{document}